\documentclass[reprint,,twocolumn,color,superscriptaddress,psfig,showpacs,amsmath,amssymb,floatfix,fleqn]{revtex4-1}

\usepackage{epsf}

\usepackage{amssymb}
\usepackage{amsmath}
\usepackage{natbib}
\setcitestyle{square,numbers}
\usepackage{amsthm}
\usepackage{graphicx}
\usepackage{dcolumn}
\usepackage{bm}
\usepackage{color}
\usepackage{subfigure}
\usepackage[subnum]{cases}
\usepackage{float}
\usepackage{epstopdf}
\newtheorem{definition}{Definition}
\newtheorem{theorem}{Theorem}
\newtheorem{lemma}{lemma}

\begin{document}


\title{Synchronization stability and circuit experiment of hyperchaos with time delay using impulse control}

\author{Hai-Peng~Ren}\thanks{Corresponding author:Hai-Peng Ren,renhaipeng@xaut.edu.cn}
 \affiliation{Shaanxi Key Laboratory of Complex System Control and Intelligent Information Processing, Xi¡¯an University of Technology, Xi¡¯an, 710048, China}
  \affiliation{School of Mechatronics Engineering, Xi'an Technological University, Xi¡¯an, 710048, China }
\author{Kun Tian}
 \affiliation{Shaanxi Key Laboratory of Complex System Control and Intelligent Information Processing, Xi¡¯an University of Technology, Xi¡¯an, 710048, China}

\author{Chao Bai}
 \affiliation{School of Mechatronics Engineering, Xi'an Technological University, Xi¡¯an, 710048, China }

\date{\today}

\begin{abstract}
Secure communication using hyperchaos has a better potential performance, but hyperchaotic impulse circuits synchronization is a challenging task. In this paper, an impulse control method is proposed for the synchronization of two hyperchaotic Chen circuits. The sufficient conditions for the synchronization of hyperchaotic systems using the impulse control are given. The upper bound of the impulse interval is derived to assure the synchronization error system to be asymptotically stable. Simulation and circuit experiment show the correctness of the analysis and feasibility of the proposed method.
\end{abstract}

\maketitle

\section{\label{Intro}Introduction}
Since Pecora and Carroll presented a pioneering work of chaos synchronization \citep{1}, chaos synchronization plays an important role in chaos control. In recent years, chaos synchronization methods have been developed rapidly, such as unidirectionally coupling method \citep{2}, OGY (Ott, Grebogi and Yorke) method \citep{3}, active control method \citep{4}, backstepping control method \citep{6}, nonlinear feedback control \citep{7}, ${H_\infty}$ method \citep{9}, adaptive feedback control \citep{10,11}, adaptive variable structure control \citep{12} and impulse control method \citep{13}. It has been studied theoretical and experimental for many chaotic systems, such as fractional-order system \citep{14}, and two identical or different chaotic systems \citep{5} and network oscillators \citep{kelly}. The application of chaos synchronization has been applied in engineering field, e.g., secure communication \citep{15,16}, motor synchronization and vibration compactor \citep{zhuanli}.

Time delay could generate an infinite dimensional hyperchaos \citep{18}, Chen system with time delay is an example demonstrating multiple kinds of attractors, including the single scroll attractor \citep{123add}, the double scroll attractors \citep{18,19} and the composite multi-scroll attractors \citep{20}, which possesses multiple positive Lyapunov exponents \citep{123add}, more complex dynamics, and better application potential. Therefore it can be applied to encryption, secure communication, etc.. However, the synchronization of the hyperchaotic systems is a challenge task.

Because the impulse control only demands for small fragment of the state variables from driving system, it permits stabilization and synchronization of chaotic systems with better energy efficiency and higher information security \citep{13,17add,18add,19add}. By employing Lyapunov-Razumikhin theorem, the global exponential stability are provided for the impulse
differential equations \citep{22}. The proof of the uniform asymptotic stability for impulse synchronization of a lot of low dimension systems are presented in previous works \citep{13a,21}. Due to a lack of effective tool and impulse delay differential equation analysis, there is rare investigation about impulse synchronization stability for the chaos system with time delay. In this paper, an impulse control method is proposed to implement hyperchaos synchronization.

The contribution of this paper is to implement the synchronization between two infinite dimensional hyperchaos given by impulse delay differential equations. On the one hand, the uniform asymptotic stability with impulse control is investigated and a sufficient condition for the impulse synchronization is derived based on the stability theory of impulsive delay differential equation. On the other hand, we verify the correctness of the analysis and the effectiveness of the method by both simulation and experiment. From our best knowledge, there is no theoretical together with experimental result about this method. Furthermore, the result is extended to the systems with multiple time delay.

The rest of the paper is organized as follows. Section 2 gives the stability analysis of the impulse synchronization. Section 4 gives the simulation and experimental results of hyperchaotic system synchronization of the single-scroll attractor in Chen system with linear time delay using the proposed method. Conclusions are given in Section 5.

\section{Impulse synchronization of time delay induced hyperchaotic attractor}
\subsection{Preliminary of impulse delay-differential equations}

In general, an impulsive differential equation is given by
\begin{equation}\label{1}
\left\{ \begin{array}{l}
 {\bf{\dot x}} = {\bf{f}}(t,\bf{x}),\begin{array}{*{20}{c}}
   {} & {} & {} & {} & {} & {} & {} & {}& {} & {} & {} & {} & {} & {}& {}\\
\end{array}\begin{array}{*{20}{c}}
   {} & {}  \\
\end{array}  \it{{t} \ne {t_k}} \\
 \Delta {\bf{x}} = {\bf{x}}(t_k^ + ) - {\bf{x}}(t_k^ - )=\bm{I_k}({\bf{x}}),{\rm{ }}\begin{array}{*{20}{c}}
   {} & {} & {} \\
\end{array}{\rm{ }}t = {t_k},k \in N \\
 {\bf{x}}(t_0^ + ) = \phi  \\
 \end{array} \right.
\end{equation}
where ${{\bf{x}}(t) \in {R^n}}$ is  state vector, ${{\bf{f}}:{R_ + } \times S\left( \rho \right) \to {R^n}}$, ${\bm{I_k}:S\left( \rho  \right) \to {R^n}}$, ${{\bf{f}}}$ is continuous function vector on ${\left( {{t_{k - 1}},{t_k}} \right] \times S\left( \rho  \right)}$, the impulse time ${{t_k}}$ satisfy ${0 < {t_0} < {t_1} < {t_2} < ...}$ and ${{\lim _{k \to\infty }}{t_k} = \infty}$.

Assuming that, for all ${k}$, ${{\bf{f}}\left( {t,0} \right) \equiv 0}$ and ${\bm{I_k}\left( 0 \right) = 0}$, equation \eqref{1} has a trivial solution. The following definitions and lemmas are introduced \citep{21}:

${{K_1} = \left\{ {g \in C\left( {{R_ + },{R_ + }} \right)|g\left( 0 \right) = 0,g\left( s \right) > 0{\kern 1pt} {\kern 1pt} {\kern 1pt} {\kern 1pt}, \forall {\kern 1pt} s > 0} \right\}}$

${{K_3} = \left\{ {g \in C\left( {{R_ + },{R_ + }} \right)|g\left( 0 \right) = 0,g\left( s \right) > 0{\kern 1pt} {\kern 1pt} , \forall{\kern 1pt} {\kern 1pt} {\kern 1pt} s > 0{\kern 1pt} } \right\}}$, ${{g\left( s \right)}}$ is a non-decreasing function.

${S\left( \rho  \right) = \left\{ {x \in {R^n}\left| {\left\| x \right\| < \rho } \right.} \right\}}$, where ${\left\|  \cdot  \right\|}$ represents ${{R^n}}$ space euclidean norm.

\begin{definition}
Let ${{V_0} = \left\{ {V:{R_ + } \times R_ + ^n \to {R_ + }} \right\},}$ ${V \in {V_0},}$  for ${(t,x(t)) \in}$ ${\left( {nT,\left( {{\rm{n}} + {\rm{1}}} \right)T} \right)}$ ${\times R_ + ^n,}$ system \eqref{1} can be described as ${\dot{\bf{x}}(t) = f(t,{\bf{x}}(t))}$, then the upper right derivative of the solution of system \eqref{1} is defined as ${D^ + }V[t,{\bf{x}}(t)]$ $ = \mathop {\lim }\limits_{h \to {0^ + }} \sup \frac{1}{h}\left\{ {V\left[ {t + h,{\bf{x}}(t) + hf\left( {t,{\bf{x}}\left( t \right)} \right)} \right] - V\left( {t,{\bf{x}}\left( t \right)} \right)} \right\}$.
\end{definition}

\begin{definition}
 Assume that ${r \in N}$, ${D \subset R}$ and ${F \subset R}$, ${PC\left( {D,F} \right)}$ denotes a piecewise continuous function from ${D}$ to ${F}$, namely, if ${\phi  \in PC\left( {D,F} \right)}$, when ${t \in D,t \ne {t_k}}$, ${\phi }$ is a continuous function, except ${t={t_k}}$, ${\phi }$ is a discontinuous function, but left side continuous. Denote ${P{C^r}\left( {D,F} \right)}$ as ${r}$-order piecewise global differentiable function from ${D}$ to ${F}$, namely if ${\phi  \in P{C^r}\left( {D,F} \right)}$, then ${\phi :D \to F,\frac{{{d^r}\phi }}{{d{t^r}}} \in PC\left( {D,F} \right)}$.
\end{definition}

\begin{lemma}
 Assume existing ${a,b,c \in {K_1}}$, ${g \in {K_3}}$, ${p \in PC\left( {{R_ + },{R_ + }} \right)}$ and  ${V:\left[ { - r,\infty } \right) \times S\left( \rho  \right)}$\\${\to {R_ + }}$, where ${V}$ is continuous on ${\left( { - r,{t_0}} \right] \times S\left( \rho  \right)}$ and
 ${\left( {{t_{k - 1}},{t_k}} \right] \times S\left( \rho  \right)}$, ${k = 1,2,...}$, for each ${x \in S\left( \rho  \right)}$, ${k = 0,1,}$ ${2,....,}$ ${{\lim _{\left( {t,y} \right) \to \left( {t_k^ - ,x} \right)}}V\left( {t,y} \right) = V\left( {t_k^ - ,x} \right)}$ exists; ${V}$ satisfying Lipschitz condition, is restricted on ${{R_ + } \times S\left( \rho  \right)}$, if the following conditions hold:

(1) ${b\left( {\left| {\bf{x}} \right|} \right) \le V\left( {t,{\bf{x}}} \right) \le a\left( {\left| {\bf{x}} \right|} \right)}$, ${\left( {t,x} \right) \in}$ ${\left[ { - r,} \right.\left. \infty  \right) \times S\left( \rho  \right)}$;

(2)${{D^ + }V(t,\phi (0)) \le p(t)c(V(t,\phi (0)))}$, for all ${t \ne {t_k}}$ in ${{R_ + }}$, and ${\phi  \in PC\left( {[ - r,0],S\left( \rho  \right)} \right)}$ whenever
${V(t,\phi (0)) \ge g(V(t + s,\phi (s)))}$ for ${s \in [ - r,0]}$, where ${g:R_ + ^n }$ ${\to R_ + ^n}$ is monotone increasing;

(3)
 ${V\left( {{t_k},\phi \left( 0 \right) + {I_k}} \right) \le g\left( {V\left( {t_{_k}^ - ,\phi \left( 0 \right)} \right)} \right)}$ for all
  ${\left( {{t_k},\phi } \right) \in {R_ + } \times PC\left. {\left( {\left[ { - r,0} \right]} \right., S\left( {{\rho _1}} \right)} \right)}$, for which ${\phi \left( {{0^ - }} \right) = \phi \left( 0 \right)}$;

(4) ${\delta  = {\sup _{k \in z}}\left\{ {{\tau _k} - {\tau _{k - 1}}} \right\} < \infty}$, where ${\delta}$ is the impulse interval:
\begin{center}
${{M_1} = {\sup _{t > 0}}\int_t^{t + \delta } {p\left( s \right)} ds < \infty}$,\\
${{M_2} = {\inf _{q > 0}}\int_{g\left( q \right)}^q {\frac{{ds}}{{c\left( s \right)}}} > {M_1}}$,\\
\end{center}
then, the  trivial solution of system \eqref{1} is uniformly asymptotically stable.
\end{lemma}

\subsection{Hyper-chaos synchronization using impulse control}
For the drive system :
\begin{equation}\label{eq12}
{\bf{\dot x}}(t) = {\bm A}{\bf{x}}\left( t \right) + {\bm B}{\bf{x}}(t - \tau ) + {\bm{\varphi }} \left( {{\bf{x}}\left( t \right)} \right),
\end{equation}
where ${{\bm A} \in {R^{n \times n}}}$ is the state matrix, ${{\bm B}}$ is the delay gain matrix, ${{\bm{\varphi }} \left( {\bf{x}} \right)}$ is the nonlinear function, ${\tau}$ is the delay time. To achieve synchronization with system \eqref{eq12}, a response system using impulse control is given by:
\begin{equation}\label{eq13}
\left\{ {\begin{array}{*{20}{c}}
   {{\bf{\dot x'}}(t) = {\bm A}{\bf{x'}}\left( t \right) + {\bm B}{\bf{x'}}(t - \tau ) + {\bm{\varphi }} '\left( {{\bf{x}}\left( t \right)} \right)} & {t \ne {t_k}}  \\
   {\Delta {\bf{x'}} = {\bf{x'}}(t_k^ + ) - {\bf{x'}}(t_k^ - ) = {{\bm{C}}_k}({\bf{x'}}({t_k}) - {\bf{x}}({t_k}))} & {t = {t_k}}  \\
\end{array}} \right.,
\end{equation}
where ${{\bf{x'}}(t)}$ represents the state vector of response system, ${\Delta {\bf{x'}}}$ represents the impulse control force at time ${t_k}$, ${{\bm {C_k}} \in {R^{n \times n}}}$ represents impulse control matrix. The remaining variables are the same equation \eqref{eq12}. From drive system \eqref{eq12} and response system \eqref{eq13}, we obtain the synchronization error system given by:

\begin{equation}\label{4}
\left\{ \begin{array}{l}
 {\bf{\dot e}}(t) = {\bm A}{\bf{e}}\left( t \right) + {\bm B}{\bf{e}}(t - \tau ) + {\bm{\varphi }} \left( {{\bf{x}}(t),{\bf{x'}}(t)} \right)\begin{array}{*{20}{c}}
   {}  \\
\end{array}t \ne {t_k} \\
 \Delta {\bf{e}}(t) = {\bm {C_k}}{\bf{e}}(t_k^ - )\begin{array}{*{20}{c}}
   {} & {} & {} & {} & {} & {} & {} & {} & {} & {} & {} & {} \\
\end{array}t = {t_k} \\
 \end{array} \right.,
\end{equation}
where ${{\bf{e}}\left( t \right) = {\bf{x}}\left( t \right) - {\bf{x'}}\left( t \right)}$, ${{\bm{\varphi }} \left( {{\bf{x}}(t),{\bf{x'}}(t)} \right) = {\bm{\varphi }} \left({{\bf{x}}(t)} \right)-}$ ${ {\bm{\varphi }} \left( {{\bf{x'}}(t)} \right)}$.

As long as the error system \eqref{4} is asymptotically stable at the origin, system \eqref{eq13} synchronize with system \eqref{eq12}. The uniform asymptotic stability theorem of error system \eqref{4} is given as follows.
\begin{theorem}
Considering the error system \eqref{4}, if the following conditions are satisfied:\\
(1) There exists a constant ${{L_1}}$ so that ${{\left\| {{\bm{\varphi }} \left( {\bf{x}\left( t \right)} \right)} \right\|^2} \le {L_1}{\left\| \bf{x} \right\|^2}}$\\
(2) ${M = \frac{{{\lambda _{\max }}\left( {{{\bm A}^T}{\bm P} + {\bm P}{\bm A}} \right) + 2{\lambda _{\max }}\left( {{{\bm P}^T}{\bm P}} \right)}}{{{\lambda _{\min }}\left( {\bm P}\right)}} +\frac{{{L_1}}}{{{\lambda _{\min }}\left( {\bm P} \right)}} + }$\\ ${\frac{{{{\left\| {\bm B} \right\|}^2}{\lambda _{\max }}\left( {\bm P} \right){\lambda _{\min }}\left( {\bm P} \right)}}{{{\lambda _{\max }}\left( {{{\left( {{\bm I} + {\bm {C_k}}} \right)}^T}{\bm P}\left( {{\bm I} + {\bm {C_k}}} \right)} \right)}} > 0}$\\
${0 < \delta  <  - \frac{{\ln \left( {{\lambda _{\max }}\left( {{{\left( {{\bm I} + {\bm {C_k}}} \right)}^T}{\bm P}\left( {{\bm I} + {\bm {C_k}}} \right)} \right)/{\lambda _{\min }}\left( {\bm P} \right)} \right)}}{M}}$,\\
where ${{\bm P}}$ is symmetric and positive definited matrix, ${{\lambda _{\max }}( \cdot )}$ is the maximum eigenvalue of the matrix, ${{\lambda _{\min }}( \cdot )}$ is the minimum eigenvalue of the matrix, and ${\delta}$ is the impulse interval.\\
Then, error system \eqref{4} is uniform asymptotically stable.\\
\end{theorem}

\begin{proof}
Select Lyapunov function candidate as:

\begin{equation}
V\left( x \right) = {{\bf{e}}^T}{\bm {P}}{\bf{e}}
\end{equation}

For ${t = {t_k}}$:

\begin{small}
\begin{equation}
\begin{array}{l}
 V\left( {{t_k},{\bf{e}}\left( {{t_k}} \right)} \right) = {{\bf{e}}^T}\left( {{t_k}} \right)P{\bf{e}}\left( {{t_k}} \right) \\
 \begin{array}{*{20}{c}}
   {} & {} & {} & {} & {} & {} & {} & {} \\
\end{array} = {\left[ {\left( {I + {C_k}} \right){\bf{e}}\left( {t_k^ - } \right)} \right]^T}P\left[ {\left( {I + {C_k}} \right){\bf{e}}\left( {t_k^ - } \right)} \right] \\
 \begin{array}{*{20}{c}}
   {} & {} & {} & {} & {} & {} & {} & {} \\
\end{array} = {\bf{e}}{\left( {t_k^ - } \right)^T}\left( {{{\left( {I + {C_k}} \right)}^T}P\left( {I + {C_k}} \right)} \right){\bf{e}}\left( {t_k^ - } \right) \\
 \begin{array}{*{20}{c}}
   {} & {} & {} & {} & {} & {} & {} & {} \\
\end{array} \le \left[ {{{{\lambda _{\max }}\left( {{{\left( {I + {C_k}} \right)}^T}P\left( {I + {C_k}} \right)} \right)} \mathord{\left/
 {\vphantom {{{\lambda _{\max }}\left( {{{\left( {I + {C_k}} \right)}^T}P\left( {I + {C_k}} \right)} \right)} {{\lambda _{\min }}\left( P \right)}}} \right.
 \kern-\nulldelimiterspace} {{\lambda _{\min }}\left( P \right)}}} \right]\\
   \begin{array}{*{20}{c}}
   {} & {} & {} & {} & {} & {} & {} & {} & {} & {}& {} & {}\\
\end{array} \cdot {\bf{e}}{\left( {t_k^ - } \right)^T}P{\bf{e}}\left( {t_k^ - } \right) \\
 \begin{array}{*{20}{c}}
   {} & {} & {} & {} & {} & {} & {} & {} \\
\end{array}  \le \left[ {{{{\lambda _{\max }}\left( {{{\left( {I + {C_k}} \right)}^T}P\left( {I + {C_k}} \right)} \right)} \mathord{\left/
 {\vphantom {{{\lambda _{\max }}\left( {{{\left( {I + {C_k}} \right)}^T}P\left( {I + {C_k}} \right)} \right)} {{\lambda _{\min }}\left( P \right)}}} \right.
 \kern-\nulldelimiterspace} {{\lambda _{\min }}\left( P \right)}}} \right]\\
  \begin{array}{*{20}{c}}
   {} & {} & {} & {} & {} & {} & {} & {} & {} & {}& {} & {} \\
\end{array} \cdot V\left( {t_k^ - ,{\bf{e}}\left( {t_k^ - } \right)} \right) \\
 \begin{array}{*{20}{c}}
   {} & {} & {} & {} & {} & {} & {} & {} \\
\end{array} = g\left( {V\left( {t_k^ - ,{\bf{e}}\left( {t_k^ - } \right)} \right)} \right) \\
 \end{array}
\end{equation}
\end{small}
here ${g\left( {\bf{V}} \right) = \frac{{{\lambda _{\max }}\left( {{{\left( {{\bm I} + {\bm {C_k}}} \right)}^T}{\bm P}\left( {{\bm I} + {\bm {C_k}}} \right)} \right)}}{{{\lambda _{\min }}\left( {\bm P} \right)}}{\bf{V}}}$.

If ${V\left( {t,{\bf{e}}\left( t \right)} \right) \ge g\left( {V\left( {t + s,{\bf{e}}\left( {t + s} \right)} \right)} \right)}$, for ${- \tau  \le s \le 0}$,
then

$\begin{array}{l}
 V\left( {t,{\bf{e}}\left( t \right)} \right) \ge
 \left( {{\lambda _{\max }}\left( {{{\left( {{\bm I} + {\bm {C_k}}} \right)}^T}{\bm P}\left( {{\bm I} + {\bm {C_k}}} \right)} \right)/{\lambda _{\min }}\left( {\bm P} \right)} \right) \cdot\\
 V\left( {t + s,{\bf{e}}\left( {t + s} \right)} \right). \\
 \end{array}$
%
%

For ${t \ne {t_k}}$:\\
\begin{equation}\label{eq8}
\begin{array}{l}
{D^ + }V({\bf{e}}) =  \\
 {{\bf{e}}^T}\left( {{A^T}P + PA} \right){\bf{e}} + 2{{\bf{e}}^T}PB{\bf{e}}\left( {t - \tau } \right) + 2{{\bf{e}}^T}P\varphi \left( {{\bf{e}}\left( t \right)} \right) \\
 \mathop {}\nolimits^{} \mathop {}\limits^{} \mathop {}\nolimits^{} \mathop {}\limits^{} \mathop {}\nolimits^{} \mathop {}\limits^{} \mathop {}\limits^{} \mathop {}\nolimits^{} \mathop {}\limits^{} \mathop {}\nolimits^{} \mathop {}\limits^{}  \le \left[ {{\lambda _{max}}\left( {{A^T}P + PA} \right)/{\lambda _{min}}\left( P \right)} \right]{{\bf{e}}^T}P{\bf{e}} \\
 \mathop {}\nolimits^{} \mathop {}\limits^{} \mathop {}\nolimits^{} \mathop {}\nolimits^{} \mathop {}\limits^{} \mathop {}\nolimits^{} \mathop {}\nolimits^{} \mathop {}\limits^{} \mathop {}\nolimits^{} \mathop {}\nolimits^{} \mathop {}\limits^{} \mathop {}\nolimits^{} \mathop {}\nolimits^{} \mathop {}\limits^{} \mathop {}\nolimits^{}  + 2{\left\| {P{\bf{e}}} \right\|^2} + {\left\| {B{\bf{e}}\left( {t - \tau } \right)} \right\|^2} + {\left\| {\varphi \left( {{\bf{e}}\left( t \right)} \right)} \right\|^2} \\
 \mathop {}\nolimits^{} \mathop {}\limits^{} \mathop {}\nolimits^{} \mathop {}\limits^{} \mathop {}\nolimits^{} \mathop {}\limits^{} \mathop {}\limits^{} \mathop {}\nolimits^{} \mathop {}\limits^{} \mathop {}\nolimits^{} \mathop {}\limits^{}  \le \left[ {{{{\lambda _{max}}\left( {{A^T}P + PA} \right)} \mathord{\left/
 {\vphantom {{{\lambda _{max}}\left( {{A^T}P + PA} \right)} {{\lambda _{min}}\left( P \right)}}} \right.
 \kern-\nulldelimiterspace} {{\lambda _{min}}\left( P \right)}}} \right]V\left( {t,{\bf{e}}\left( t \right)} \right) \\
 \mathop {}\nolimits^{} \mathop {}\limits^{} \mathop {}\nolimits^{} \mathop {}\nolimits^{} \mathop {}\limits^{} \mathop {}\nolimits^{} \mathop {}\nolimits^{} \mathop {}\limits^{} \mathop {}\nolimits^{} \mathop {}\nolimits^{} \mathop {}\limits^{} \mathop {}\nolimits^{} \mathop {}\nolimits^{} \mathop {}\limits^{} \mathop {}\nolimits^{}  + 2{\left\| {P{\bf{e}}} \right\|^2} + {\left\| B \right\|^2}{\left\| {{\bf{e}}\left( {t - \tau } \right)} \right\|^2} + {L_1}{\left\| {{\bf{e}}\left( t \right)} \right\|^2} \\
 \mathop {}\nolimits^{} \mathop {}\limits^{} \mathop {}\nolimits^{} \mathop {}\limits^{} \mathop {}\nolimits^{} \mathop {}\limits^{} \mathop {}\limits^{} \mathop {}\nolimits^{} \mathop {}\limits^{} \mathop {}\nolimits^{} \mathop {}\limits^{}  \le \left[ {{{{\lambda _{max}}\left( {{A^T}P + PA} \right)} \mathord{\left/
 {\vphantom {{{\lambda _{max}}\left( {{A^T}P + PA} \right)} {{\lambda _{min}}\left( P \right)}}} \right.
 \kern-\nulldelimiterspace} {{\lambda _{min}}\left( P \right)}}} \right]V\left( {t,{\bf{e}}\left( t \right)} \right) \\
 \mathop {}\nolimits^{} \mathop {}\limits^{} \mathop {}\nolimits^{} \mathop {}\nolimits^{} \mathop {}\limits^{} \mathop {}\nolimits^{} \mathop {}\nolimits^{} \mathop {}\limits^{} \mathop {}\nolimits^{} \mathop {}\nolimits^{} \mathop {}\limits^{} \mathop {}\nolimits^{} \mathop {}\nolimits^{} \mathop {}\limits^{} \mathop {}\nolimits^{}  + 2\left[ {{{{\lambda _{max}}\left( {{P^T}P} \right)} \mathord{\left/
 {\vphantom {{{\lambda _{max}}\left( {{P^T}P} \right)} {{\lambda _{min}}\left( P \right)}}} \right.
 \kern-\nulldelimiterspace} {{\lambda _{min}}\left( P \right)}}} \right]V\left( {t,{\bf{e}}\left( t \right)} \right) \\
 \mathop {}\nolimits^{} \mathop {}\limits^{} \mathop {}\nolimits^{} \mathop {}\nolimits^{} \mathop {}\limits^{} \mathop {}\nolimits^{} \mathop {}\nolimits^{} \mathop {}\limits^{} \mathop {}\nolimits^{} \mathop {}\nolimits^{} \mathop {}\limits^{} \mathop {}\nolimits^{} \mathop {}\nolimits^{} \mathop {}\limits^{} \mathop {}\nolimits^{}  + \left[ {{{{L_1}V\left( {t,{\bf{e}}\left( t \right)} \right)} \mathord{\left/
 {\vphantom {{{L_1}V\left( {t,{\bf{e}}\left( t \right)} \right)} {{\lambda _{min}}\left( P \right)}}} \right.
 \kern-\nulldelimiterspace} {{\lambda _{min}}\left( P \right)}}} \right] + {\left\| B \right\|^2}{\left\| {{\bf{e}}\left( {t - \tau } \right)} \right\|^2} \\
 \mathop {}\nolimits^{} \mathop {}\limits^{} \mathop {}\nolimits^{} \mathop {}\limits^{} \mathop {}\nolimits^{} \mathop {}\limits^{} \mathop {}\limits^{} \mathop {}\nolimits^{} \mathop {}\limits^{} \mathop {}\nolimits^{} \mathop {}\limits^{}  \le \left[ {\left( {{\lambda _{max}}\left( {{A^T}P + PA} \right) + 2{\lambda _{max}}\left( {{P^T}P} \right) + {L_1}} \right)} \right. \\
 \mathop {}\nolimits^{} \mathop {}\limits^{} \mathop {}\nolimits^{} \mathop {}\nolimits^{} \mathop {}\limits^{} \mathop {}\nolimits^{} \mathop {}\nolimits^{} \mathop {}\limits^{} \mathop {}\nolimits^{} \mathop {}\nolimits^{} \mathop {}\limits^{} \mathop {}\nolimits^{} \mathop {}\nolimits^{} \mathop {}\limits^{} \mathop {}\nolimits^{} \left. {{{} \mathord{\left/
 {\vphantom {{} {{\lambda _{min}}\left( P \right)}}} \right.
 \kern-\nulldelimiterspace} {{\lambda _{min}}\left( P \right)}}} \right] \cdot V\left( {t,e\left( t \right)} \right) + {\left\| B \right\|^2}{\left\| {e\left( {t - \tau } \right)} \right\|^2} \\
 \mathop {}\nolimits^{} \mathop {}\limits^{} \mathop {}\nolimits^{} \mathop {}\limits^{} \mathop {}\nolimits^{} \mathop {}\limits^{} \mathop {}\limits^{} \mathop {}\nolimits^{} \mathop {}\limits^{} \mathop {}\nolimits^{} \mathop {}\limits^{}  \le p(t) \cdot V\left( {t,e\left( t \right)} \right) \\
 \end{array}
\end{equation}

\begin{flushleft}
where ${
p\left( t \right) = \frac{{{\lambda _{\max }}\left( {{A^T}P + PA} \right) + 2{\lambda _{\max }}\left( {{P^T}P} \right) + {L_1}}}{{{\lambda _{\min }}\left( P \right)}}
}$
\end{flushleft}

\begin{center}
${
+ \frac{{{{\left\| B \right\|}^2}{\lambda _{\max }}\left( P \right){\lambda _{\min }}\left( P \right)}}{{{\lambda _{\max }}\left( {{{\left( {I + {C_k}} \right)}^T}P\left( {I + {C_k}} \right)} \right)}}}$
\end{center}

assume ${c\left( s \right) = s,p\left( t \right) = M}$

According to condition (4) of Lemma 1, we have:

\begin{equation}
\begin{array}{l}
 {M_2} - {M_1} = \mathop {\inf }\limits_{q > 0} \smallint _{g\left( q \right)}^q\frac{{ds}}{{c\left( s \right)}} - \mathop {\sup }\limits_{t > 0} \int\limits_t^{t + \delta } {p\left( s \right)ds}  \\
  \mathop {}\nolimits^{} \mathop {}\nolimits^{} \mathop {}\nolimits^{} \mathop {}\nolimits^{} \mathop {}\nolimits^{} \mathop {}\nolimits^{} \mathop {}\nolimits^{} \mathop {}\nolimits^{} \mathop {}\nolimits^{} \mathop {}\nolimits^{} \mathop {}\nolimits^{} \mathop {}\nolimits^{} \mathop {}\nolimits^{} \mathop {}\nolimits^{} \mathop {}\nolimits^{} \mathop {}\nolimits^{} \mathop {}\nolimits^{} \mathop {}\nolimits^{} \mathop {}\nolimits^{} \mathop {}\nolimits^{}  = \ln q - \ln g\left( q \right) - M\delta  \\
   \mathop {}\nolimits^{} \mathop {}\nolimits^{} \mathop {}\nolimits^{} \mathop {}\nolimits^{} \mathop {}\nolimits^{} \mathop {}\nolimits^{} \mathop {}\nolimits^{} \mathop {}\nolimits^{} \mathop {}\nolimits^{} \mathop {}\nolimits^{} \mathop {}\nolimits^{} \mathop {}\nolimits^{} \mathop {}\nolimits^{} \mathop {}\nolimits^{} \mathop {}\nolimits^{} \mathop {}\nolimits^{} \mathop {}\nolimits^{} \mathop {}\nolimits^{} \mathop {}\nolimits^{} \mathop {}\nolimits^{}
   =- \ln \frac{{{\lambda _{\max }}\left( {{{\left( {{\bm I} + {\bm {C_k}}} \right)}^T}{\bm P}\left( {{\bm I} + {\bm {C_k}}} \right)} \right)}}{{{\lambda _{\min }}\left( \bm{P} \right)}} - M\delta  > 0
\end{array}
\end{equation}

According to condition (2) of Theorem 1, if the following conditions are established:

\begin{equation}
\begin{array}{l}
 M = \frac{{{\lambda _{\max }}\left( {{A^T}P + PA} \right) + 2{\lambda _{\max }}\left( {{P^T}P} \right)}}{{{\lambda _{\min }}\left( P \right)}} + \frac{{{L_1}}}{{{\lambda _{\min }}\left( P \right)}} \\
 \end{array}
\end{equation}

${
 + \frac{{{{\left\| B \right\|}^2}{\lambda _{\max }}\left( P \right){\lambda _{\min }}\left( P \right)}}{{{\lambda _{\max }}\left( {{{\left( {I + {C_k}} \right)}^T}P\left( {I + {C_k}} \right)} \right)}} > 0.
}$

${\begin{array}{*{20}{c}}
   {} & {}  \\
\end{array} }$

${0 < \delta  <  - \frac{{\ln \left( {{\lambda _{\max }}\left( {{{\left( {I + {C_k}} \right)}^T}P\left( {I + {C_k}} \right)} \right)/{\lambda _{\min }}\left( P \right)} \right)}}{M}.}$

Then, we have a conclusion that the error system is asymptotically stability.
\end{proof}

\noindent
\textbf{Corollary} \emph{Theorem 1 can be extended to the delay differential equations with multiple time delay for different linear items.}

It is clear that conditions of Theorem 1 is independent of ${\tau}$. We present Theorem 2 about the system with multiple different time delays.

Considering the system

\begin{equation}\label{eqr14}
{\bf{\dot x}}(t) = {\bm A}{\bf{x}}(t) + {\bm B}{\bf{x}}(t - \tau ) + {\bm{\varphi }} ({\bf{x}}(t),{\bf{x}}(t - \tau ))
\end{equation}

\noindent which contains two different time delays, then we have

\begin{theorem}
Considering error system of the equation (\ref{eqr14}), if the following conditions are satisfied:\\
(1) There exist constants ${{L_1}}$ and ${L_2}$ so that ${{\left\| {{\bm{\varphi }} ({\bf x}(t),{\bf x}(t - \tau ))} \right\|^2} \le}$ ${ {L_1}{\left\| {{\bf x}(t)} \right\|^2} + {L_2}{\left\| {{\bf x}(t - \tau )} \right\|^2}}$\\
(2) ${M = {\lambda _{\max }}({{\bm A}^T}{\bm P} + {\bm P}{\bm A})/{\lambda _{\min }}({\bm P})+2{\lambda _{\max }}({{\bm P}^T}{\bm P})/{\lambda _{\min }}({\bm P}) + }$  ${{L_1}/{\lambda _{\min }}({\bm P}){\rm{  + }}({\left\| {\bm B} \right\|^2} + {L_2}){\lambda _{\max }}({\bm P}){\lambda _{\min }}({\bm P})/{\lambda _{\max }}({({\bm I} + {\bm {C_k}})^T}{\bm P}}$ ${({\bm I} + {\bm {C_k}}))}$ \\
${0 < \delta  < \frac{{ - \ln ({\lambda _{\max }}({{({\bm I} + {\bm {C_k}})}^T}{\bm P}({\bm I} + {\bm {C_k}}))/{\lambda _{\min }}({\bm P}))}}{M}}$,\\
where matrix ${{\bm P}}$ is positive definite, then the error system is uniform asymptotically stable.\\
\end{theorem}

\begin{proof}
Lyapunov function candidate is selected as

\begin{equation}\label{12}
V\left( x \right) = {{\bf{e}}^T}{\bm P}{\bf{e}}
\end{equation}

The time derivation of \eqref{12} at ${t={t_k}}$ is the same as Eq. \eqref{eq8}.

For ${t \ne {t_k}}$,

\begin{equation}
\begin{aligned}
\begin{array}{l}
 {D^ + }V({\bm e}) = {{\bm e}^T}({{\bm A}^T}{\bm P} + {\bm P}{\bm A}){\bm e} + 2{{\bm e}^T}{\bm P}{\bm B}{\bf x}(t - \tau ) + \\
 2{{\bm e}^T}{\bm P}{\bm{\varphi }} ({\bm e}(t),{\bm e}(t - \tau )) \\
 \begin{array}{*{20}{c}}
\end{array}  \le ({\lambda _{\max }}({{\bm A}^T}{\bm P} + {\bm P}{\bm A})/{\lambda _{\min }}({\bm P}) +
 2{\lambda _{\max }}({{\bm P}^T}{\bm P})/{\lambda _{\min }}({\bm P}) + \\{L_1}/{\lambda _{\min }}({\bm P})) \cdot V(t,{\bm e}(t))
  + ({\left\| {\bm B} \right\|^2} + {L_2}){\left\| {{\bf x}(t - \tau )} \right\|^2} \\
 \begin{array}{*{20}{c}}
\end{array} \le p(t)V(t,{\bm e}(t)) \\
 \end{array}
\end{aligned}
\end{equation}
where
${p(t) = {\lambda _{\max }}({{\bm A}^T}{\bm P} + {\bm P}{\bm A})/{\lambda _{\min }}({\bm P}) + 2{\lambda _{\max }}({{\bm P}^T}{\bm P})/{\lambda _{\min }}({\bm P})}$ ${+{L_1}/{\lambda _{\min }}({\bm P}) + \\ ({\left\| {\bm B} \right\|^2} + {L_2}){\lambda _{\max }}({\bm P}){\lambda _{\min }}({\bm P})/({\lambda _{\max }}(({\bm I} + {\bm {C_k}}{)^T}{\bm P}({\bm I} + {\bm {C_k}})) > 0 \\
}$
denote ${c(s) = s,p(t) = M}$.

If the following
conditions are established

\begin{center}
\begin{small}
${\begin{array}{l}
 M = {\lambda _{\max }}({{\bm A}^T}{\bm P} + {\bm P}{\bm A})/{\lambda _{\min }}({\bm P}) + 2{\lambda _{\max }}({{\bm P}^T}{\bm P})/{\lambda _{\min }}({\bm P})
  + {L_1}/{\lambda _{\min }}({\bm P} {\rm{ + }}({\left\| {\bm B} \right\|^2} + {L_2}){\lambda _{\max }}({\bm P}){\lambda _{\min }}({\bm P})/{\lambda _{\max }}(({\bm I} + {\bm {C_k}}{)^T}{\bm P}({\bm I} + {\bm {C_k}}))
 \end{array}
}$
\end{small}
\end{center}
\begin{center}
${\begin{array}{l}
0 < \delta  < \frac{{ - \ln ({\lambda _{\max }}({{({\bm I} + {\bm {C_k}})}^T}{\bm P}({\bm I} + {\bm {C_k}}))/{\lambda _{\min }}({\bm P}))}}{M} \\
 \end{array}
}$
\end{center}

Then, with a similar process as the single time delay system, we have a conclusion that the error system of equation (\ref{eqr14}) is asymptotically stability.
\end{proof}

\section{Simulation and circuit experiment results}

\subsection{Hyper-chaotic single-scroll attractor in Chen system with time-delay}
\label{sec:1}
The Chen system is given as follows:

\begin{equation}
\left\{ \begin{array}{l}
 \dot x(t) = a(y(t) - x(t)); \\
 \dot y(t) = (c - a)x(t) - x(t)z(t) + cy(t); \\
 \dot z(t) = x(t)y(t) - bz(t).
\end{array} \right.
\end{equation}

When the parameters ${a = 35,b = 3,c = 18.35978}$, Chen system is non-chaotic. Chen system with linear time delay feedback is given as follows \citep{19}:

\begin{equation}
\left\{ \begin{array}{l}
 \dot x(t) = a(y(t) - x(t)); \\
 \dot y(t) = (c - a)x(t) - x(t)z(t) + cy(t); \\
 \dot z(t) = x(t)y(t) - bz(t) + k(z(t) - z(t - \tau )) \\
 \end{array} \right.
\end{equation}

When the parameters ${k = 2.85,\tau  = 0.3}$, the initial value is ${x(0)=2.27}$, ${y(0)=2.27}$, ${z(0)=1.72}$, ${z(t)=0}$ in ${t \in [-0.3,0)}$, the system exhibits a single-scroll attractor \citep{123add}.

\subsection{Simulation results}
Chen system with linear time delay feedback control can be given as:

\begin{equation}\label{eq1}
{\bf{\dot x}}(t) = {\bm A}{\bf{x}}\left( t \right) + {\bm B}{\bf{x}}(t - \tau ) + {\bm{\varphi }} \left( {{\bf{x}}\left( t \right)} \right),
\end{equation}
where
${{\bf{x}} = \left[ {\begin{array}{*{20}{c}}
   x  \\
   y  \\
   z  \\
\end{array}} \right]}$,
${{\bm A} = \left[ {\begin{array}{*{20}{c}}
   { - 35} & {35} & 0  \\
   { - 17} & {18} & 0  \\
   0 & 0 & {-3+K}  \\
\end{array}} \right]}$,
${{\bm B} = \left[ {\begin{array}{*{20}{c}}
   0 & 0 & 0  \\
   0 & 0 & 0  \\
   0 & 0 & -K  \\
\end{array}} \right]}$,
${{\bm{\varphi }} ({\bf{x}}(t)) = \left[ {\begin{array}{*{20}{c}}
   0  \\
   { - x(t)z(t)}  \\
   {x(t)y(t)}  \\
\end{array}} \right]}$, ${K = 3.8}$, ${\tau  = 0.3}$.

The response system with impulse control is given as follows:
\begin{small}
\begin{equation}\label{eq17q}
\left\{ {\begin{array}{*{20}{c}}
   {{\bf{\dot x'}}(t) = {\bm A}{\bf{x'}}\left( t \right) + {\bm B}{\bf{x'}}(t - \tau ) + {\bm{\varphi }} \left( {{\bf{x'}}\left( t \right)} \right)} & {t \ne {t_k}}  \\
   {\Delta {\bf{x'}} = {\bf{x'}}\left( {t_k^ + } \right) - {\bf{x'}}\left( {t_k^ - } \right) = {\bm C}\left( {{\bf{x'}}\left( {{t_k}} \right) - {\bf{x}}\left( {{t_k}} \right)} \right)} & {t = {t_k}}  \\
\end{array}} \right.,
\end{equation}
\end{small}
where
${{\bm{\varphi }} \left( {{\bf{x'}}\left( t \right)} \right) = \left[ {\begin{array}{*{20}{c}}
   0  \\
   { - x'z'}  \\
   {x'y'}  \\
\end{array}} \right]}$, ${{\bf{x'}} = \left[ {\begin{array}{*{20}{c}}
   {x'}  \\
   {y'}  \\
   {z'}  \\
\end{array}} \right]}$, \\
${\bm A}$ and ${\bm B}$ are the same as that in Eq. (\ref{eq1}).

The error system is given by:
\begin{equation}\label{eq18as}
\left\{ {\begin{array}{*{20}{c}}
   {{\bf{\dot e}}(t) = {\bm A}{\bf{e}}\left( t \right) + {\bm B}{\bf{e}}(t - \tau ) + {\bm{\varphi }} \left( {{\bf{x}},{\bf{x'}}} \right)} & {t \ne {t_k}}  \\
   {\Delta {\bf{e}} = {\bm C}{\bf{e}}\left( {t_k^ - } \right)} & {t = {t_k}}  \\
\end{array}} \right.,
\end{equation}
where ${{\bf{e}}(t) = {\bf{x}}(t) - {\bf{x'}}(t)}$, ${{\bm{\varphi }} \left( {{\bf x},{\bf x}'} \right) = {\bm{\varphi }} \left( {\bf x} \right) - {\bm{\varphi }} \left( {{\bf x}'} \right)}$, ${\bm A}$, ${\bm B}$ are the same as that in Eq. (\ref{eq1}).

Impulse control matrix ${
{\bm C} = \left( {\begin{array}{*{20}{c}}
   { - 0.2} & 0 & 0  \\
   0 & { - 1.8} & 0  \\
   0 & 0 & { - 1.8}  \\
\end{array}} \right)}$, for convenience, let ${{\bf{P}} = I}$.

${{\left\| {{\bm{\varphi }} \left( {{\bf{x}}\left( t \right)} \right)} \right\|^2} \le {L_1}{\left\| {\bf{x}} \right\|^2} = 16{\left\| {\bf{x}} \right\|^2}
}$, thus condition (1) in Theorem 1 is satisfied.

\begin{figure}[!t]
\centering
\includegraphics[width=3.5in]{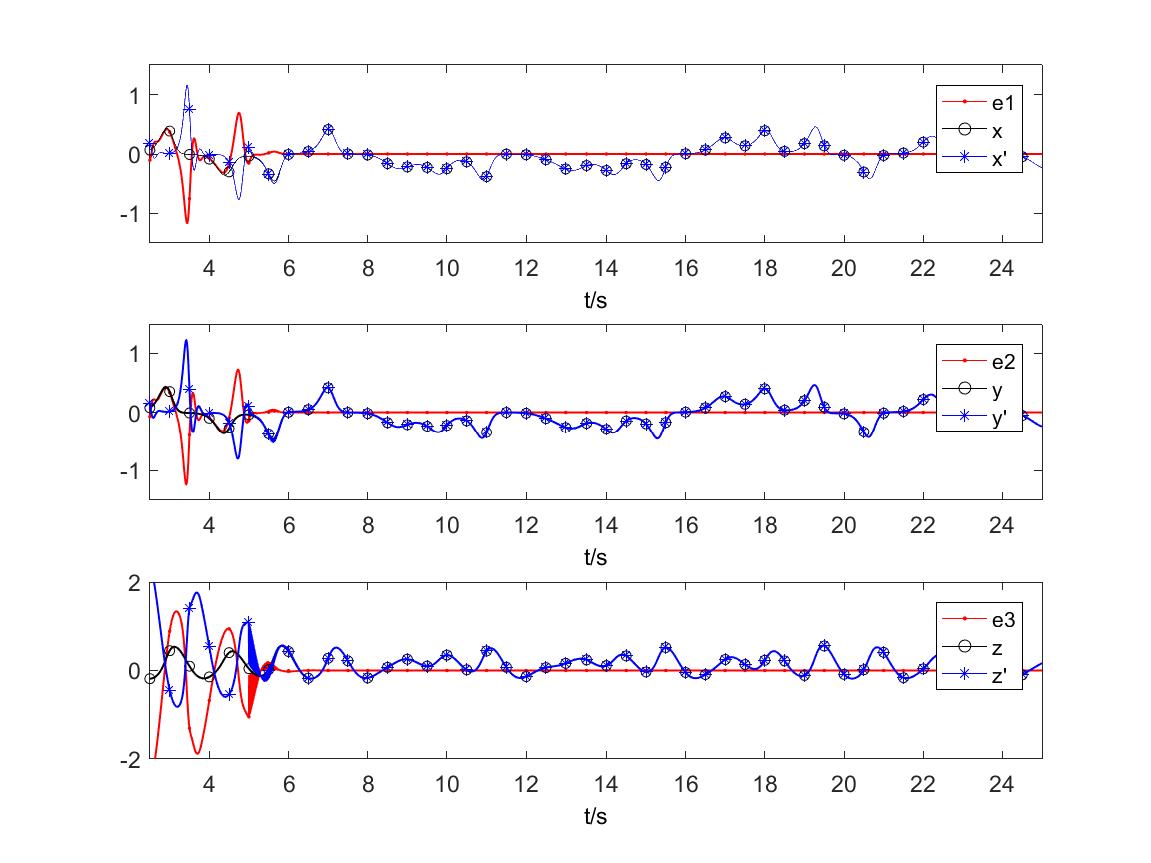}
\caption{ Simulation waveform of synchronization of the delay Chen system using impulse control.}
\label{fig_sim6}
\end{figure}

From condition (2) in Theorem 1, we have\\
(1) ${M = \frac{{{\lambda _{\max }}\left( {{A^T}P + PA} \right) + 2{\lambda _{\max }}\left( {{P^T}P} \right) + {L_1}}}{{{\lambda _{\min }}\left( P \right)}}}$

${\begin{array}{*{20}{c}}
   {} & {} & {} & {}  \\
\end{array}}$

${ + \frac{{{{\left\| B \right\|}^2}{\lambda _{\max }}\left( P \right){\lambda _{\min }}\left( P \right)}}{{{\lambda _{\max }}\left( {{{\left( {I + {C_k}} \right)}^T}P\left( {I + {C_k}} \right)} \right)}}
= \frac{{38.9732 + 2 + 16}}{1} + \frac{{14.44}}{{0.64}}= 79.5357 > 0}$  \\

(2)${0 < \delta  <  - \frac{{\ln \left( {{\lambda _{\max }}\left( {{{\left( {I + {C_k}} \right)}^T}P\left( {I + {C_k}} \right)} \right)/{\lambda _{\min }}\left( P \right)} \right)}}{M} \\ }$ \\
${\begin{array}{*{20}{c}}
   {} & {} & {} & {}  \\
\end{array}}$

${ =  - \frac{{\ln \left( {0.64} \right)}}{{79.5357}} = 0.005611 \\}$.

If we choose the impulse interval ${\delta   = 0.005 < 0.005611}$, then all conditions in Theorem 1 are satisfied. The simulation results using impulse control for ${\delta =0.005}$ are given in Fig. \ref{fig_sim6}, the impulse control is activated from ${t=5s}$. The subplot (a) gives the curves of states $x$ of drive system, $x'$ of response system and the state error ${{e_1}\left( t \right)}$ in black line with circle, blue line with star and red line with dot, respectively. From Fig. \ref{fig_sim6}, we know that the synchronization is achieved after the impulse controller is active at $t=5s$.

\subsection{Circuit experimental results}

The schematic diagram of impulse control circuit is given in Fig. \ref{fig_sim}.

\begin{figure}[h]
\centering
\includegraphics[width=3in]{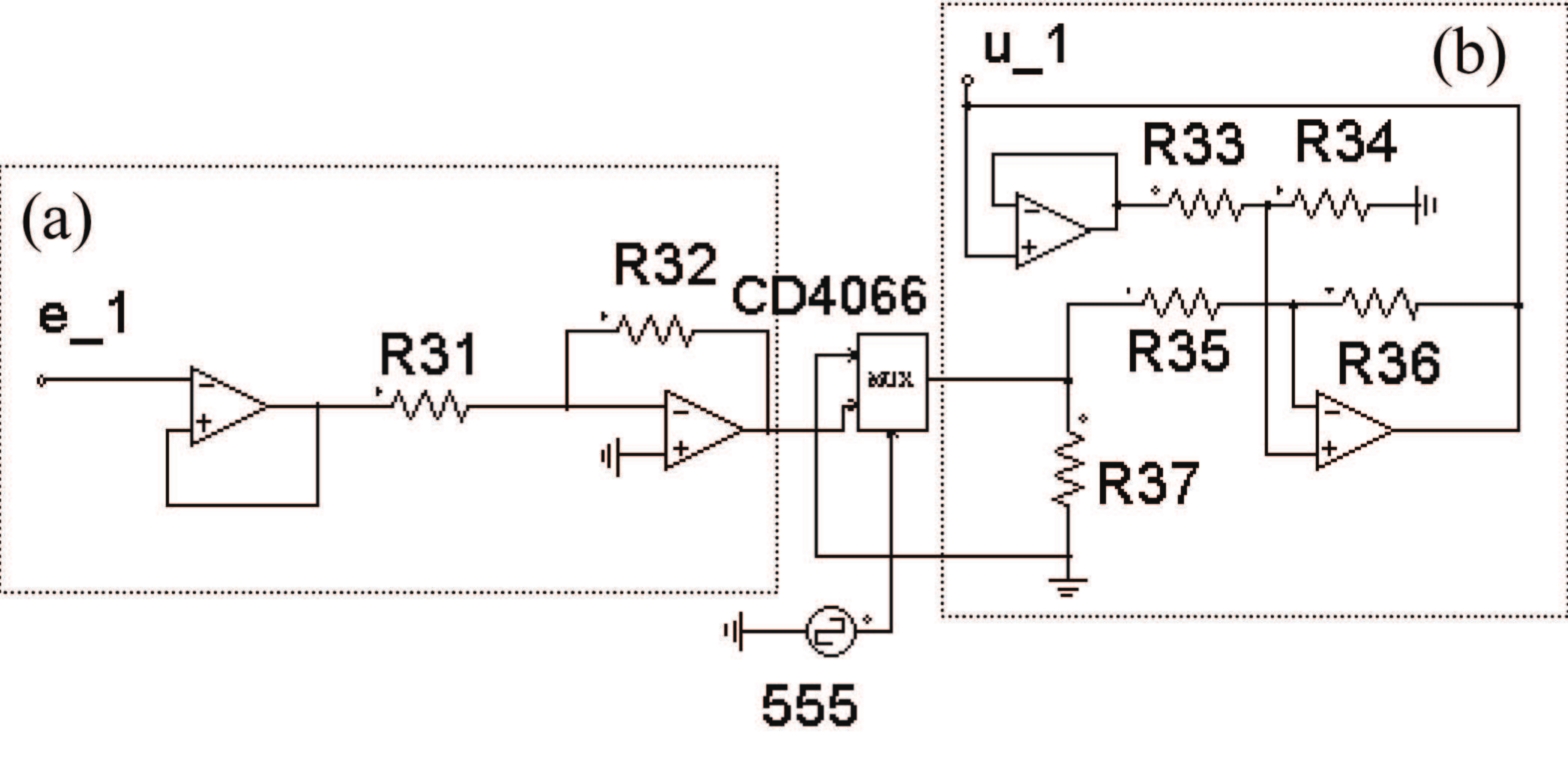}
\caption{ Schematic diagram of the impulse control.}
\label{fig_sim}
\end{figure}

In Fig. \ref{fig_sim}, block (a) is mainly used for the proportional amplification. Error ${e_1} = x - x'$ passes through the proportional amplifier, getting proportional gain ${{c_1} =  - \frac{{{R_{32}}}}{{{R_{31}}}}{\kern 1pt} {\kern 1pt}}$.

The 555 timer provides impulse signal with 10\% duty ratio and 0.005s impulse interval. The Fig. \ref{fig_sim}(a) is regard to the impulse gain $\bm C$ in Eq. (\ref{eq18as}), and the Fig. \ref{fig_sim} is used for signal power amplification. The CD4066 in the Fig. \ref{fig_sim} is employed as the multipath selector, when the switch is on-stated by a high level.


The schematic diagram of the impulse synchronization is given in Fig. \ref{fig_sim2}. In Fig. \ref{fig_sim2}, signal $x$, $y$ and $z$ are the state variables of the drive system, and ${x'}$, ${y'}$ and ${z'}$ are the state variables of the response system. Let ${
{\bf{e}} = {({e_1},{e_2},{e_3})^T}}$ be the synchronization error. The impulse control outputs are ${u_1}$, ${u_2}$ and ${u_3}$. The circuit parameters are summerized in Table 1.

\begin{figure*}[ht]
\centering
\includegraphics[width=5.5in]{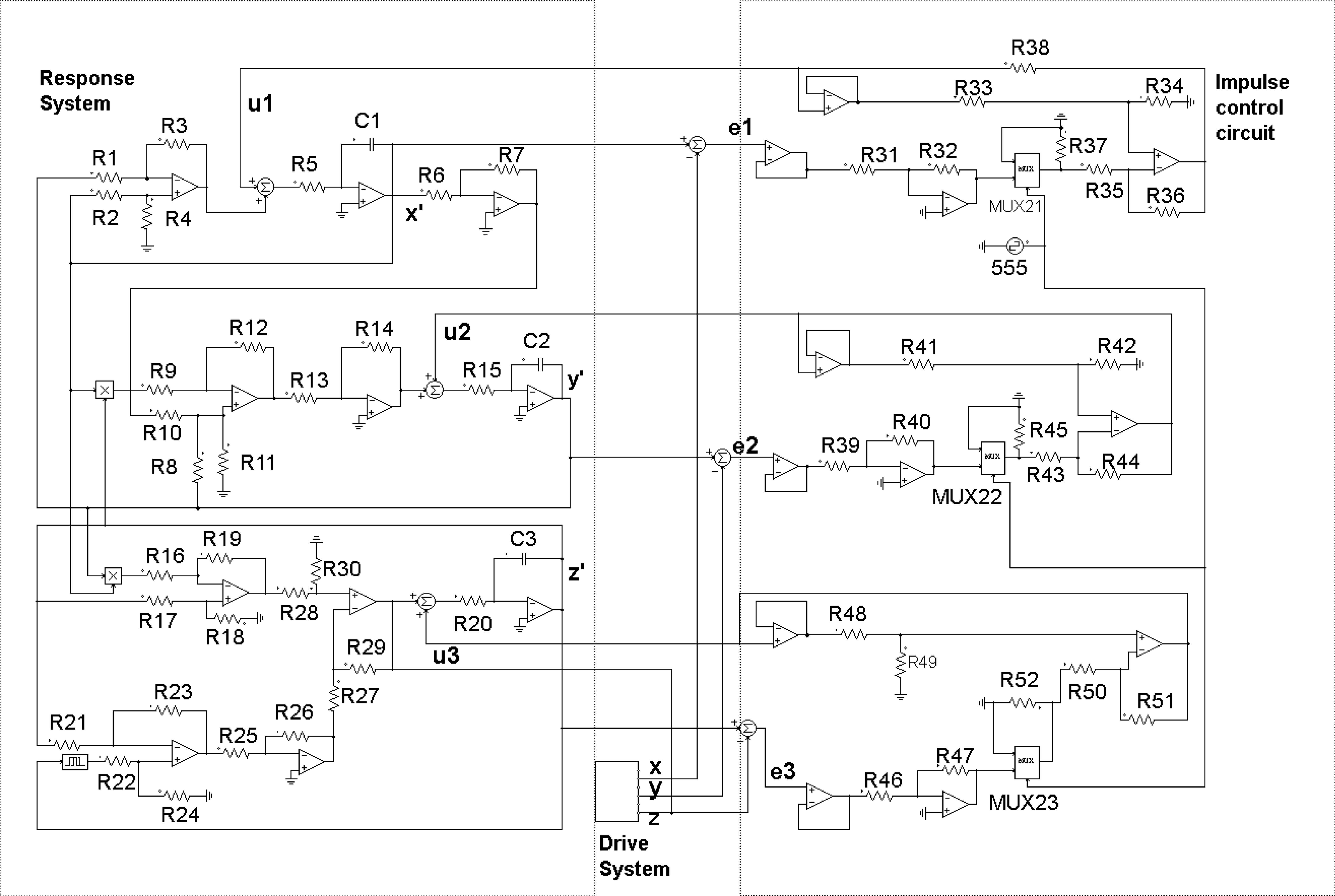}
\caption{The schematic diagram of the impulse control for hyperchaotic Chen circuits synchronization.}
\label{fig_sim2}
\end{figure*}

\begin{table}
\caption{Circuit Component values.}
\centering
\begin{tabular}{l|l}
 Component                    & Value         \\
   \hline
  \hline
  R1, R2, R3, R4, R6, R7, R13, R14, R16, R17,\\
  R19, R21, R22, R23, R24, R27, R28,R29, R30 & 10$k\Omega$\\
  R5, R57                         & 2.86$M\Omega$\\
  R8                          & 16.5$k\Omega$   \\
  R9                      & 3$k\Omega$\\
  R10                       &18.25$k\Omega$ \\
  R11                   & 60$k\Omega$ \\
  R12, R25, R31, R33, R34, R35, R36, R39, R41,\\
  R42, R43, R44, R46, R48, R49, R50, R51                 & 1$k\Omega$ \\
  R15                          & 3.33$M\Omega$ \\
  R18                         & 1.765$k\Omega$ \\
  R20                         & 10$M\Omega$ \\
  R26                         & 285$\Omega$ \\
  R37, R45, R52                      & 12$\Omega$ \\
  R38                          & 100$\Omega$ \\
  R40,R47                       & 1.8$k\Omega$ \\
  R45                          & 0.12$k\Omega$ \\
  C1, C2, C3, C4, C5, C6             & 10000$PF$\\
\end{tabular}
\end{table}

The circuit simulation results (using PSIM software) of the synchronization between the drive system and the response system are given in Fig. \ref{fig4}. The upper, middle and bottom subplot in Fig. \ref{fig4} are the corresponding synchronization error curves $e_1$, $e_2$, $e_3$, respectively. Here, the impulse control is activated after $t=80$ second. We can see, from Fig. \ref{fig4}, that the synchronization error tend to zero after the controller is activated.

In the following, the circuit experiment is built to observe the synchronization between the two hyperchaotic systems. The operation amplifiers used in circuits are LF347N, the multipliers used are AD633JN, the inverter used is 74LS04, the multi-channel selector is CD4066. The experimental results is shown in Fig. \ref{fig_sim11}.

\begin{figure}[t]
\centering
\includegraphics[width=3.5in]{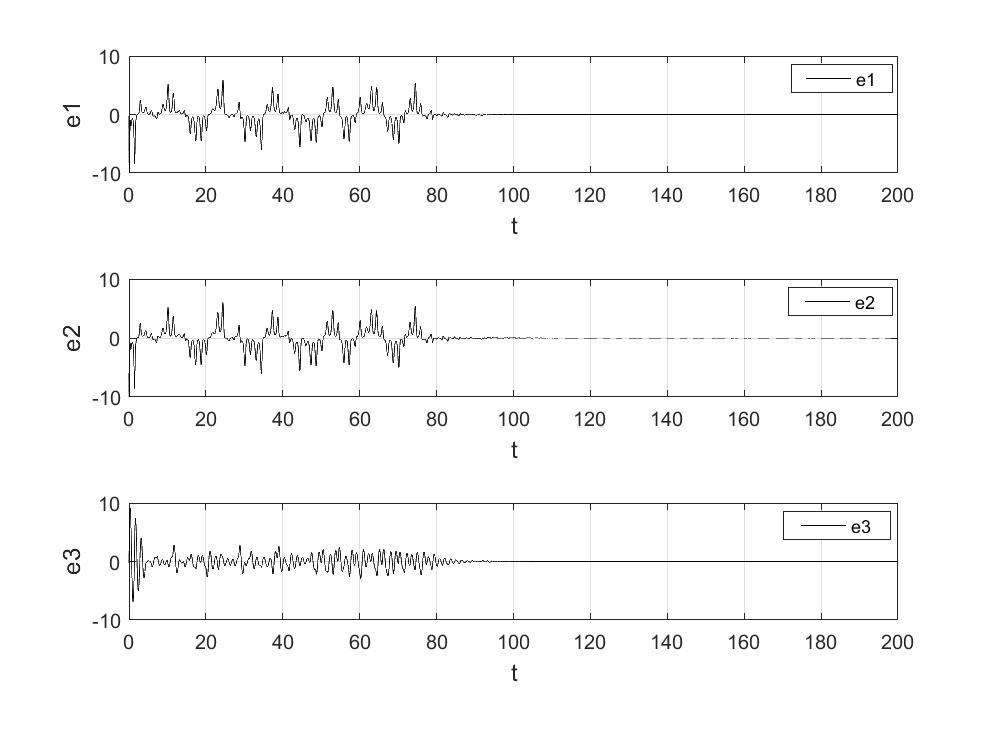}
\caption{ Simulation waveform of synchronization of the delay Chen system using impulse control.}
\label{fig4}
\end{figure}

Subplot (a), (c) and (e) in Fig. \ref{fig_sim11} represent the states of the two systems without the impulse control and the corresponding synchronization errors, before impulse control is put into effect; subplot (b), (d) and (f) represent the states of the two systems with the impulse control and the corresponding synchronization errors after the controller is put into effect. In Fig. \ref{fig_sim11} (a) and (b), channel 1 represents ${x}$ waveform of the drive system,  channel 2 represents ${x'}$ waveform of the response system, channel 4 represents ${{e_1}}$; In Fig. \ref{fig_sim11}(c) and (d), channel 3 represents ${y}$ waveform of the drive system,  channel 2 represents ${y'}$ waveform of the response system, channel 4 represents ${{e_2}}$; In Fig. \ref{fig_sim11} (e) and (f), channel 3 represents ${z}$ waveform of the drive system, channel 1 represents ${z'}$ waveform of the response system, channel 4 represents ${{e_3}}$; From Fig \ref{fig_sim11}, we learn that the synchronization is achieved after the impulse control is put into effect.

\begin{figure}[!t]
\centering
\includegraphics[width=3.4in]{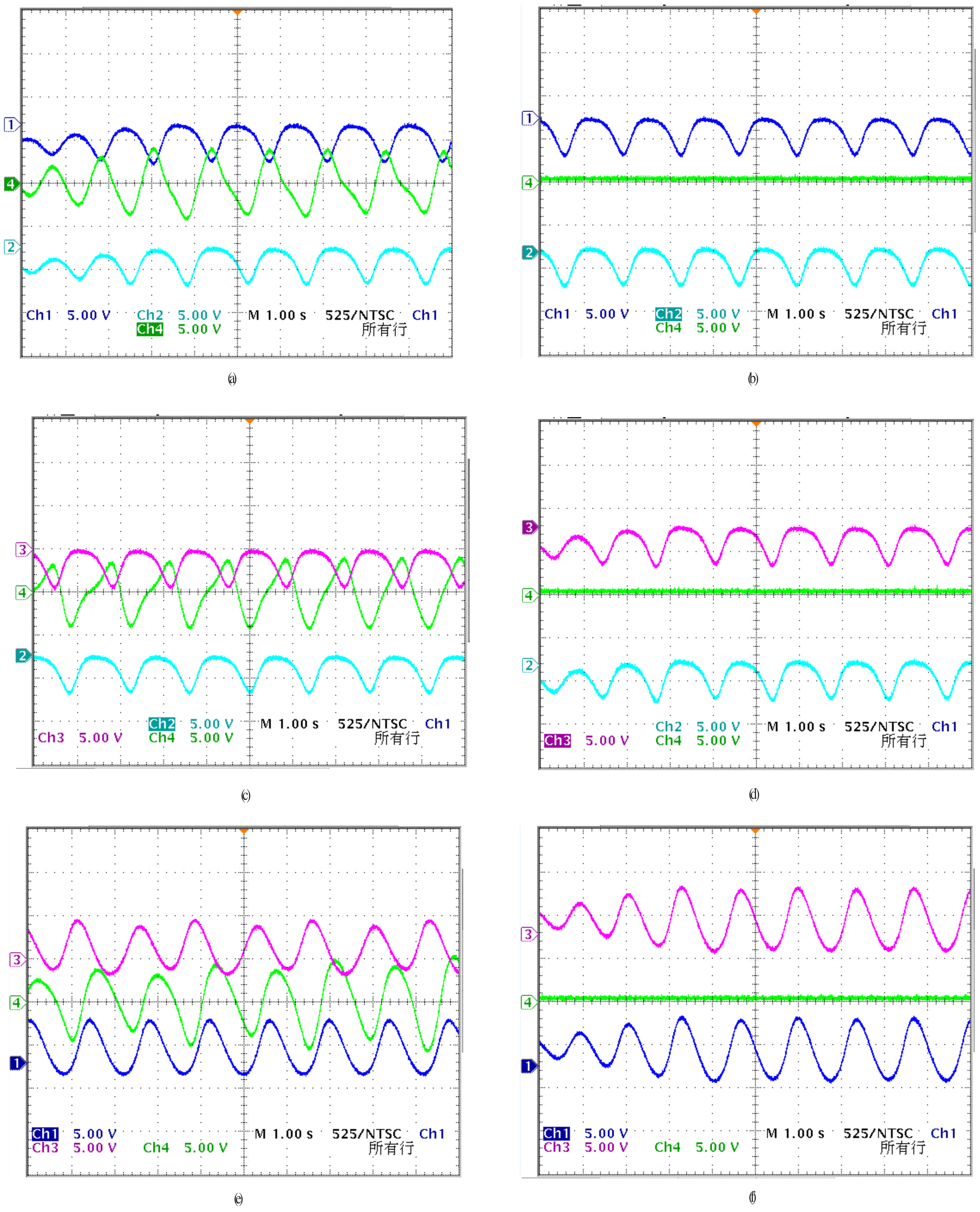}
\caption{ The experimental waveforms.}
\label{fig_sim11}
\end{figure}

\section{Conclusion}

To conclude, in this paper, with the Lyapunov stability theory of impulse delay differential equation, we propose the sufficient condition for the impulse synchronization of the systems with single (multiple) time delay feedback item(s). The proposed method has been applied to the Chen system with time delay over simulation and electrical circuit, which have shown its correctness and effectiveness of the theory. It is worth noticing that the impulse control does contributes to the secure communication, clinical treatment and species population. This method gives better prospects to explore some real world systems.

\section*{Acknowledgement(s)}

The work is supported in part by National Natural Science Foundation of China (60804040, 61172070) and Shaanxi Province Innovation Research Team (2013KCT-04).

%

%

\end{document}